\documentclass[leqno,sumlimits,intlimits,namelimits,draft]{amsart}

\usepackage[T1]{fontenc}

\usepackage[numeric,initials,bibtex-style,nobysame]{amsrefs}

\usepackage[utf8]{inputenc}
\usepackage[USenglish]{babel}

\usepackage{amscd}
\usepackage[all]{xy}
\usepackage{eucal}
\usepackage{amsthm}

\usepackage{mathrsfs}
\usepackage{bm}

\usepackage{ifthen}
\usepackage{fancybox}
\usepackage{fancyhdr}

\usepackage{pifont}

\newcommand{\diag}{\ensuremath{\mathrm{diag}}}

\newcommand{\mb}[1]{\ensuremath{\mathbb{#1}}}

\newcommand{\R}{\mb{R}}

\newcommand{\al}{\alpha}
\newcommand{\be}{\beta}
\newcommand{\ga}{\gamma}
\newcommand{\Ga}{\Gamma}

\newcommand{\eps}{\varepsilon}

\newcommand{\vphi}{\varphi}

\newcommand{\la}{\lambda}

\newcommand{\sig}{\sigma}

\newcommand{\D}{\ensuremath{\mathcal{D}}}
\newcommand{\G}{\ensuremath{\mathcal{G}}}
\renewcommand{\S}{\mathscr{S}}

\newcommand{\loc}{\ensuremath{\text{loc}}}

\renewcommand{\d}{\ensuremath{\partial}}
\newcommand{\diff}[1]{\frac{d}{d#1}}

\newcommand{\isom}{\cong}
\newcommand{\col}{\colon}

\newcommand{\norm}[2]{{\| #1 \|}_{#2}}

\newcommand{\Norm}[1]{\norm{#1}{}}

\newcommand{\FT}[1]{\widehat{#1}}

\newcommand{\beq}{\begin{equation}}
\newcommand{\eeq}{\end{equation}}

\renewcommand{\Re}{\ensuremath{\mathop{\mathrm{Re}}}}
\renewcommand{\Im}{\ensuremath{\mathop{\mathrm{Im}}}}

\theoremstyle{definition}
\newtheorem{theorem}{Theorem}[section]
\newtheorem{lemma}[theorem]{Lemma}
\newtheorem{proposition}[theorem]{Proposition}
\newtheorem{definition}[theorem]{Definition}

\newtheorem{remark}[theorem]{Remark}
\newtheorem{example}[theorem]{Example}

\newcommand{\co}{\ensuremath{{[0,1]}}}
\newcommand{\tauco}{\ensuremath{\tau_{\text{c.o.}}}}

\begin{document}

\title{Conical spacetimes and global hyperbolicity}

\author{G\"unther H\"ormann\\
Fakult\"at f\"ur Mathematik,
Universit\"at Wien, Austria}

\date{\today}

\begin{abstract}  Vickers and Wilson (\cite{VW:00}) have shown global hyperbolicity of the conical spacetime in the sense of well-posedness of the initial value problem for the wave equation in generalized functions. We add the aspect of metric splitting and preliminary thoughts on Cauchy hypersurfaces and causal curves. 
\end{abstract}

\maketitle

\pagestyle{plain}

\section{Introduction and set-up of the cosmic string model}

Our basic references are \cite{Wald:84} for general relativity, \cite{ONeill:83} for smooth differential and Lorentz\-ian geometry, and \cite{GKOS:01} for generalized functions and non-smooth differential geometry. A \emph{spacetime} $(M,g)$ shall consist of a connected smooth manifold $M$ with a symmetric covariant $2$-tensor field $g$ that is measurable (with respect to some [hence any] Lebesgue measure on $M$ in the sense of \cite[Section 16.22]{Die:72}), locally bounded (with respect to some [hence any] Riemannian metric on $M$), and is almost everywhere non-degenerate and of index $1$ (dimension of a maximal subspace in the tangent space where the metric is negative definite). Furthermore, $(M,g)$ shall be \emph{time-oriented} by the existence and choice of a nowhere vanishing continuous vector field that is timelike almost everywhere. 

We consider a simple model of a \emph{cosmic string} in terms of the so-called \emph{conical spacetime} (cf.\ \cite[Section 3.4]{GP:09} and \cite[Section 7]{SV:06}), the metric of which with parameter $0 < \al < 1$ can be specified in cylindrical coordinates $(t,r,\vphi,z) \in \R \times\, ]0,\infty[\, \times\,]-\pi,\pi[\, \times \R =: W$ in the form
$$
    l_\al(t,r,\vphi,z) := - dt^2 + dr^2 + \al^2 r^2 d\vphi^2 + dz^2.
$$
It arises from a construction with the universal covering space upon removing from Minkowski space the fixed point set of rotation by the angle $2 \pi \al$ in the $(x,y)$-plane.  In  \cite[Section 2]{Vickers:87} James Vickers discussed the more general geometric situation producing the class of cosmic string models based on methods for `elementary' quasi-regular singularities by Ellis and Schmidt in \cite[Appendix A]{ES:77}. 

The above local form of the metric clearly possesses a continuous extension $\overline{l_\al}$ to $r = 0$ (and to arbitrary values of $\vphi$), which is degenerate at $r = 0$. The smooth transformation $\Phi \col (t,r,\vphi,z) \mapsto (t,x,y,z) := (t,r \cos \vphi, r \sin \vphi, z)$ is a diffeomorphism of $W$ onto $V := \Phi(W) = \{ (t,x,y,z) \in \R^4 \col (x,y) \neq (0,0) \}$, 
hence we obtain the corresponding Cartesian representation of the Lorentz metric $l_\al$ on the open dense region $V \subset \R^4$ as push-forward $\Phi_* l_\al$  with matrix  $$G_\al(t,x,y,z) =  
    \left( \begin{smallmatrix} -1 & 0 & 0 & 0\\ 
      0 & \frac{x^2 + \al^2 y^2}{x^2 + y^2} &
      (1 - \al^2) \frac{x\, y}{x^2 + y^2} & 0\\
      0 & (1 - \al^2) \frac{x \, y}{x^2 + y^2} &
      \frac{\al^2 x^2 + y^2}{x^2 + y^2} & 0\\
      0 & 0 & 0 & 1
   \end{smallmatrix} \right)
$$
that has smooth and bounded components on $V$, which define unique classes in $L^\infty(\R^4)$, since $\R^4 \setminus V$ has Lebesgue measure $0$. Hence this $L^\infty$-extension of $G_\al$ to $\R^4$ is the matrix of a metric $g_\al$. The ``pull-back'' $\Phi^* g_\al$, where   $\Phi$ is considered as map $\R \times [0,\infty[\, \times\,[-\pi,\pi]\, \times \R \to \R^4$, gives  $\overline{l_\al}$ in the sense that $D\Phi(t,r,\vphi,z)^t \cdot G_\al(\Phi(t,r,\vphi,z)) \cdot D\Phi(t,r,\vphi,z) = \diag(-1,1,\al^2 r^2,1)$ holds almost everywhere, namely on $W$.

Thus, in the Cartesian description we have $M = \R^4$ and the Lorentz metric $g_\al$  is given in global coordinates $(t,x,y,z)$ by
\begin{multline*}
   g_\al(t,x,y,z) := - dt^2 +  
   \left(\frac{1 + \al^2}{2} + \frac{1 - \al^2}{2} f_1(x,y)\right) dx^2  +
   (1 - \al^2) f_2(x,y)\, dx\, dy \\
    + \left(\frac{1 + \al^2}{2} - \frac{1 - \al^2}{2} f_1(x,y)\right) dy^2  + dz^2
\end{multline*}
with $f_1, f_2 \in L^\infty(\R^2)$ represented by following functions on $\R^2\setminus\{(0,0)\}$: 
$$  f_1(x,y) := \frac{x^2 - y^2}{x^2 + y^2} = \Re \frac{(x+iy)^2}{x^2 + y^2}  \quad\text{and}\quad 
    f_2(x,y) := \frac{2 x y}{x^2 + y^2} = \Im \frac{(x+iy)^2}{x^2 + y^2}.
$$
The eigenvalues of the matrix representing $g_\al$ at any point $(t,x,y,z)$ with $(x,y) \neq (0,0)$ are $-1$, $1$ (double), and $\al^2$, thus showing non-degeneracy in the tangent spaces over these points. As time-orientation of the conical spacetime we declare the (constant) timelike vector field $(1,0,0,0)$ to be future-directed.

We may rewrite $g_\al$ by splitting off the non-smooth part in the form
\begin{multline*}
   g_\al(t,x,y,z) := \underbrace{- dt^2 +  
   \frac{1 + \al^2}{2}\, dx^2 + \frac{1 + \al^2}{2} \, dy^2
    + dz^2}_{\eta_\al(t,x,y,z)}\\
   + \frac{1 - \al^2}{2}\,  \underbrace{\left( f_1(x,y)\, dx^2 - 
       f_1(x,y)\, dy^2 + 2 f_2(x,y)\, dx\, dy \right)}_{h(x,y)}. 
\end{multline*}
Note that $\lim_{\al \to 1} g_\al = \lim_{\al \to 1} \eta_\al$ is the Minkowski metric and $h$ can be considered a discontinuous $L^\infty$ Lorentz metric on $\R^2$, since its corresponding matrix in the $2$-dimensional tangent spaces at each $(x,y) \neq (0,0)$ has eigenvalues $1$ and $-1$. 

James Vickers has studied the conical spacetime and also identified it as the `prototype' of a more general class of spacetimes with $2$-dimensional quasi-regular singularities already in his doctoral thesis and later in a series of research articles (cf.\ \cites{Vickers:87,Vickers:90} and references therein). Geometric and distributional methods have been employed to argue that these models satisfy the Einstein equations with certain components of the energy-momentum tensor being proportional to a delta distribution in $(x,y)$-subspace, or, in fact, the Lorentz metric producing distributional curvature components of such type. Clarke, Vickers, and Wilson then in \cite{CVW:96} for the first time set-up the model of the conical spacetime within the framework of Colombeau's generalized functions and confirmed Vickers' earlier result by investigating the distributional shadow of the generalized curvature tensor.

In \cite[Section 2]{SV:06} Steinbauer and Vickers discuss an analysis going back to Geroch and Traschen which strives for `minimal sufficient conditions' on the metric components to guarantee a well-defined distributional curvature (employing multiplication in spaces of measurable functions as nonlinear operations). In this setting, the notion of a so-called \emph{gt-regular metric} is introduced, requiring the components of the metric $g$ with respect to coordinates to be in $L^\infty_\loc \cap H^1_\loc$. As noted in \cite{SV:06}  the conical spacetime metric is not gt-regular, but it narrowly it fails to be so.

\begin{lemma} The first-order distributional derivatives of $f_1$ and $f_2$ belong to $L^1_\loc \setminus L^2_\loc$, thus, $f_1, f_2 \in (L^\infty_\loc \cap W^{1,1}_\loc)\setminus  H^1_\loc$.
\end{lemma}
\begin{proof} It suffices to argue for $f_1, f_2$ as functions (or distributions) on $\R^2$, since they are constant with respect to the variables $t$ and $z$. We observe that both, $f_1$ and $f_2$, are homogeneous of order $0$, hence the first-order derivatives $\d_j f_k$ ($1 \leq j, k \leq 2$) are homogeneous distributions of order $-1$. Moreover, due to smoothness on $\R^2 \setminus \{0\}$, the derivatives there are given by the classical pointwise formulae, which define locally integrable functions (due to the homogeneity degree $-1$ in two dimensions, or as can be seen from direct inspection). By unique extension of a homogeneous distribution on $\R^n \setminus \{0\}$ to $\R^n$, if the degree of homogeneity is not an integer $\leq -n$ (cf.\ \cite[Theorem 3.2.3]{Hoermander:V1}), we obtain that the distributions $\d_j f_k$ ($1 \leq j, k \leq 2$) are given by these formulae as functions in $L^1_\loc$. On the other hand, since each $|\d_j f_k|^2 \neq 0$ and is homogeneous of order $-2$ on $\R^2$,  $\d_j f_k$ cannot be square-integrable over compact neighborhoods of $0$.
\end{proof}

As a consequence of $g$ being not gt-regular it is natural to embed the metric into Colombeau-generalized tensor fields and analyze the conical spacetime model as one with a generalized Lorentzian metric in the sense of \cite[Subsection 3.2.5]{GKOS:01} employing the concept of the so-called special Colombeau algebras on manifolds. By abuse of notation, we think of $g_\al$ as being represented by a net $(g_\al^\eps)_{\eps\in\,]0,1]}$ of covariant two-tensor fields (with components obtained by mollification---in fact, it suffices to regularize $f_1$ and $f_2$ in this way), such that on any relatively compact subset $g_\al^\eps$ is a smooth Lorentz metric when $\eps$ is sufficiently small. In \cite{VW:00} Vickers and Wilson have shown that the conical spacetime may be called \emph{globally hyperbolic} in the sense of \cite{Clarke:98}, namely in terms of well-posedness for the wave equation. It is our aim here to add a few observations and aspects of \emph{generalized global hyperbolicity} to their fundamental result.

\section{$\Box$-global hyperbolicity}

Clarke advocated the view that a singularity of a non-smooth spacetime, i.e., with $g$ below regularity $C^{1,1}$, should be manifest as an obstruction to the Cauchy development of the physical fields in that spacetime  -- contrary to the standard notion as an ``obstruction to extending geodesics'' (cf.\ \cite{Clarke:98}).
A spacetime $(M,g)$ is therefore called \emph{$\Box$-globally hyperbolic}, if there exists a $C^1$ spacelike hypersurface $S \subset M$ with unit future pointing (continuous) normal vector field $\xi$ such that $M \setminus S$ is the disjoint union of two open (connected) subsets $M^+$ and $M^-$, $\xi$ points from $M^-$ to $M^+$, and the Cauchy problem  
$$
    \Box_g\, u = f, \quad u\!\mid_S = v,\quad \xi u\!\mid_S  = w,
$$
is well-posed in appropriate Sobolev spaces. In case of a merely $L^\infty$ Lorentz metric, initial data $(v,w) \in H^1(S)\times L^2(S)$ and $f = 0$, say, the distributional solution $u$ could not be expected to be better than $H^1$ (at least in ``spatial directions'') and hence $\Box_g u$ does not even make sense due to ill-defined distributional products involving derivatives of metric components with derivatives of $u$. 

For the conical spacetime metric  Vickers and Wilson in \cite{VW:00} interpreted the model in the framework of Colombeau-generalized functions and, based on delicate energy estimates, proved well-posedness of the generalized Cauchy problem with data on $S$ given by $t = 0$, thus showing that the cosmic string spacetime is \emph{generalized $\Box$-globally hyperbolic}. Moreover, they succeed in showing that in case of embedded Cauchy data from $(v,w) \in H^1(S)\times L^2(S)$ the Colombeau solution to the generalized Cauchy problem possesses a distributional shadow in $\D'(M)$ and they were able to identify its analytic and geometric structure to quite some detail.

A further remarkable fact about \cite{VW:00} is that the well-posedness result of Vickers and Wilson for the conical spacetime is still not covered by successive theoretical development on the Cauchy problem for wave equations in non-smooth spacetimes as in \cites{GMS:09,HKS:12}.

\section{Globally hyperbolic metric splitting}

Let us take here as a basis for the notion of a smooth \emph{globally hyperbolic} spacetime the variant defined in \cite[Section 6.6]{HE:73} and \cite[Definition 5.24]{Penrose:72}, namely, the requirements of strong causality and compactness of all intersections of causal futures and causal pasts for arbitrary pairs of points in the manifold.
Among several equivalent ways to characterize global hyperbolicity of smooth spacetimes that have been established, Bernal and S\'{a}nchez in \cite{BS:05} gave the most detailed ``normal form description'' in terms of the following global metric splitting statement. 

\begin{theorem} $(M,g)$ is globally hyperbolic if and only if it is isometric with a spacetime $(\R\times S, \la)$ such that: 

\noindent (a) each $\{t\} \times S$ is a (smooth spacelike) Cauchy hypersurface,

\noindent (b) $\la(t,x) = - \theta(t,x)\, dt^2 + \rho_t(x)$ with a smoothly parametrized family $(\rho_t)_{t\in\R}$  of  Riemann metrics on $S$, and

\noindent (c) $\theta \in C^\infty(\R\times S)$ and positive.

\end{theorem}

We  show  that the conical spacetime satisfies non-smooth versions of this kind of metric splitting in both senses, one directly with $L^\infty$ Riemann metrics $\rho_t$ in (b) as above, and also in the sense of Colombeau-generalized spacetimes. 

\subsubsection*{Metric splitting for the conical spacetime in the $L^\infty$-setting} 
The natural guess for an $L^\infty$-type metric splitting in case of the conical spacetime $M = \R \times \R^3$ is, of course, to put $S = \R^3$, the isometry being the identity map, and $\la = g_\al = - dt^2 + \rho$ with ($t$-independent) $\rho$ defined by
\begin{multline*}
   \rho(x,y,z) :=  
   \left(\frac{1 + \al^2}{2} + \frac{1 - \al^2}{2} f_1(x,y)\right)\, dx^2 +
   \left(\frac{1 + \al^2}{2} - \frac{1 - \al^2}{2} f_1(x,y)\right)\, dy^2\\ 
   +  (1 - \al^2) f_2(x,y)\, dx\, dy + dz^2.
\end{multline*}
We will check that $\rho$ is an $L^\infty$ Riemann metric on $\R^3$ and that $\{t\} \times \R^3$ is a Cauchy hypersurface\footnote{Note that the condition `spacelike' will in general not make sense, since a hypersurface is a zero set for any Lebesgue measure on the manifold and the $L^\infty$ Lorentz metric might be ``blind'' to this subset.} for $g_\al$ in an appropriate sense\footnote{Namely, defined as usual by requiring that timelike inextendible curves meet the surface exactly once, but with some caution about the notion of \emph{timelike}: With general $L^\infty$ metrics we cannot be sure that the composition of a Lebesgue measurable metric component with a continuous curve tangent is measurable. However, $g_\al$ is Borel measurable, since $f_1$ and $f_2$ are Borel measurable; this simplifies the situation, because then $s \mapsto g_\al(\ga(s))(\dot{\ga}(s),\dot{\ga}(s))$ is Borel measurable for any $C^1$ curve $\ga$.}:

\noindent 1. Boundedness of $\rho$ is clear, hence it suffices to show that $\rho$ is a smooth Riemann metric on $\R^3 \setminus \{(0,0,z) \col z \in \R\}$, which is clear from the following estimate for any tangent vector $(v_1,v_2,v_3) \in T_{(x,y,z)} S \isom \R^3$
\begin{multline*}
  \rho (v,v) = \frac{1 + \al^2}{2} (v_1^2 + v_2^2) +
     \frac{1 - \al^2}{2} (f_1 v_1^2 + 2 f_2 v_1 v_2 - f_1 v_2^2) + v_3^2 \\
     = v_3^2 +\frac{1 + \al^2}{2} (v_1^2 + v_2^2)  +
       \frac{1 - \al^2}{2(x^2 + y^2)} 
       \underbrace{\big( (x v_1 + y v_2)^2 - (y v_1 - x v_2)^2 \big)}_%
         {\qquad\geq - (y v_1 - x v_2)^2 \geq -(x^2 + y^2)(v_1^2 + v_2^2)}\\
      \geq  v_3^2 + \big(\frac{1 + \al^2}{2} -  \frac{1 - \al^2}{2} \big)(v_1^2 + v_2^2)
      =  \al^2 (v_1^2 + v_2^2) + v_3^2.
\end{multline*}
 
\noindent 2. $\{t\} \times \R^3$ is a Cauchy hypersurface: Let $\ga = (\ga_0,\ga_1)$ be a $C^1$ curve that is timelike, meaning $\dot\ga(s) \neq 0$ for every $s$ and $g_\al(\ga(s))(\dot{\ga}(s),\dot{\ga}(s)) < 0$ for almost every $s$, and inextendible;  let $\Norm{\;}$ denote the Euclidean norm on $\R^3$, then we have by the above estimate 
$  0 > g_\al(\dot\ga,\dot\ga) = - \dot\ga_0^{\,2} + \rho(\dot\ga_1,\dot\ga_1) \geq
     - \dot\ga_0^{\,2} + \al^2 \Norm{\dot\ga_1}^2$ (almost everywhere),
hence $|\ga_0(s) - \ga_0(s_0)| >  \al |\int_{s_0}^s \Norm{\dot\ga_1(\tau)}\, d\tau |$ if $s \neq s_0$. If $\ga_0$ was bounded above or below then the Euclidean length of $\ga_1$ were bounded, hence $\ga_0$ (by monotonicity) as well as $\ga_1$ would have endpoints, i.e., $\ga$ could not be inextendible. If $\ga_0$ is unbounded above and below, then the image of $\ga_0$ is $\R$, hence the curve $\ga$ must hit $\{t\} \times \R^3$ at least once, and by strict monotonicity of $\ga_0$ at most once.

\subsubsection*{Metric splitting in the Colombeau-setting}
Considering now the conical spacetime as a generalized Lorentz manifold, we show that it also qualifies for an appropriate variant of the metric splitting. The corresponding generalization of this notion has been introduced in \cite{HKS:12} and been shown to provide a sufficient condition for well-posedness of the Cauchy problem for the wave equation; \cite{HS:14} discusses a slightly extended definition of the metric splitting, in order to allow for more general isometries that were needed to cover models of discontinuous wave-type spacetimes. 

\begin{definition}\label{G_metric_splitting}  Let $g$ be a generalized Lorentz metric on the smooth $(n + 1)$-dimensional manifold M. We say that $(M, g)$ allows a
\emph{globally hyperbolic metric splitting} if there exists a generalized diffeomorphism $\Phi \col M \to \R\times S$, where $S$ is an $n$-dimensional smooth manifold such that the following holds for the pushed forward generalized Lorentz
metric $\lambda := \Phi_*g$ on $\R\times S$:
\begin{trivlist}

\item[(a)] There is a representative $(\lambda_\eps)_\eps$ of $\lambda$ such that every $\lambda_\eps$ is a Lorentz metric and each
slice $\{t_0\}\times S$ with arbitrary $t_0\in\R$ is a (smooth, spacelike) Cauchy hypersurface for every $\lambda_\eps$.

\item[(b)]  We have the metric splitting of $\lambda$ in the form
$$
     \lambda=-\theta dt^2 + H,
$$
where $H\in\Gamma_\G(\mathrm{pr}_2^*(T_2^0 S))$ is a $t$-dependent generalized Riemannian metric and $\theta\in\G(\R\times
S)$ is globally bounded and locally uniformly positive, i.e., for some (hence any) representative $(\theta_\eps)_\eps$ of
$\theta$ and for every relatively compact subset $K \subset S$ we can find a constant $C > 0$ such that $\theta_\eps(t,x) \geq C$ holds
for small $\eps>0$ and $x \in K$.
  
\item[(c)]  For every $T > 0$ there exists a representative $(H_\eps)_\eps$ of $H$ and a smooth complete Riemannian metric $h$ on $S$ which uniformly bounds $H$ from below in the following sense: for all $t\in [-T, T]$, $x\in
S$, $v \in T_x S$, and $\eps\in \, ]0,1]$
$$
   (H_\eps)_t(v,v)\geq h(v,v).
$$
\end{trivlist} 
\end{definition}

\begin{remark} In fact, condition (a) follows from (b) and (c), though it might be more instructive for the notion itself to have it still explicitly in the list of basic properties. To show the redundancy of (a) we argue as follows to prove that $\{t_0\}\times S$ is a Cauchy hypersurface for every $\lambda_\eps$: Timelike curves $\ga = (\ga_0,\ga_1)$ with respect to $\lambda_\eps$ are characterized by 
$0 > (\la_\eps)_{(\ga(s))}(\dot\ga(s),\dot\ga(s)) = 
     - \theta_\eps(\ga(s)) \, \dot\ga_0(s)^2 + 
         (H_\eps)_{\ga_0(s)}(\dot\ga_1(s),\dot\ga_1(s))$,
which implies that $s \mapsto \ga_0(s)$ is strictly monotonic, since ${\dot\ga_0}^{\,2} > H_\eps(\dot\ga_1,\dot\ga_1) / \theta_\eps(\ga) > 0$; if $\ga$ is supposed to be timelike inextendible, then either $\ga_0$ is unbounded (above and below) or $\ga_1$ has no endpoints; if $\ga_0$ were bounded above or below, then (by monotonicity) $\dot\ga_0 \to 0$ towards an interval boundary of the curve, hence $H_\eps(\dot\ga_1,\dot\ga_1) \to 0$ and therefore $\dot\ga_1(s) = 0$ for every $s$ near the interval boundary by (c), which implies that $\ga_1$ has an endpoint; thus, $\ga_0$ is unbounded above and below, i.e., there exists $s_0$ in the curve interval such that $\ga_0(s_0) = t_0$, so that $\ga(s) \in \{t_0\}\times S$. By strict monotonicity of $\ga_0$, the curve hits $\{t_0\}\times S$ at most once.
\end{remark}

For the generalized function version of the conical spacetime we try with a splitting of the metric representative in the form $g_\al^\eps = - dt^2 + \rho^\eps$ with 
\begin{multline*}
   \rho^\eps(x,y,z) :=  
   \left(\frac{1 + \al^2}{2} + \frac{1 - \al^2}{2} f_1^\eps(x,y)\right)\, dx^2 +
   \left(\frac{1 + \al^2}{2} - \frac{1 - \al^2}{2} f_1^\eps(x,y)\right)\, dy^2\\ 
   +  (1 - \al^2) f_2^\eps(x,y)\, dx\, dy + dz^2.
\end{multline*}
Clearly, (b) is satisfied with $S = \R^3$, $\Phi$ the identity map, and $\theta = 1$. 

We will identify conditions on the regularization and/or on $\al$ for the following to hold: $\exists \be > 0$, $\forall (x,y,z) \in \R^3$, $\forall v = (v_1,v_2,v_3) \in T_{(x,y,z)} S \isom \R^3$
\begin{equation}\label{lowerbound}
  \rho^\eps(x,y,z)(v,v) \geq \be (v_1^2 + v_2^2) + v_3^2 \qquad
     0 < \eps \leq 1.
\end{equation} 
Then condition (c) will hold as well. Thus, the explicit globally hyperbolic metric splitting of the conical spacetime will be established, once \eqref{lowerbound} is guaranteed. 

We observe that $ \rho^\eps(x,y,z)(v,v) =$
\vspace{-1\baselineskip}
\begin{multline*}
  v_3^2 + \frac{1 + \al^2}{2}(v_1^2 + v_2^2) +
     \frac{1 - \al^2}{2} (v_1,v_2) \cdot 
     \overbrace{\begin{pmatrix} f_1^\eps(x,y) & f_2^\eps(x,y)\\ 
         f_2^\eps(x,y) & - f_1^\eps(x,y)  \end{pmatrix}}^{:= F_\eps(x,y)}
     \cdot \begin{pmatrix} v_1\\ v_2 \end{pmatrix}\\
     \geq v_3^2 + \frac{1 + \al^2}{2}(v_1^2 + v_2^2) +
     \frac{1 - \al^2}{2} \mu_\eps(x,y) (v_1^2 + v_2^2),
\end{multline*}
where $\mu_\eps(x,y)$ is the smaller eigenvalue of $F_\eps(x,y)$. Since $\det(F_\eps(x,y) - \mu I_2) = \mu^2 - f_1^\eps(x,y)^2 - f_2^\eps(x,y)^2$, we have $\mu_\eps(x,y) = - |f_1^\eps(x,y) + i f_2^\eps(x,y)|$. Hence, our final task is to provide a good upper bound for $|f_1^\eps + i f_2^\eps|$. Note that $f_1^\eps + i f_2^\eps$ is obtained from $f_1 + i f_2$ by convolution regularization with a $\delta$-net in the form $\phi_\eps(x,y) := \phi(x/\eps,y/\eps)/\eps^2$, where $\phi \in \S(\R^2)$ is real-valued (to ensure real components of $g_\al^\eps$) with $\int \phi = 1$; hence
$\norm{f_1^\eps + i f_2^\eps}{L^\infty} \leq \norm{f_1 + i f_2}{L^\infty} 
     \norm{\phi_\eps}{L^1} = 1 \cdot \norm{\phi}{L^1}$ and  
     $\mu_\eps(x,y) \geq -\norm{\phi}{L^1}$, therefore 
$\rho^\eps(x,y,z)(v,v) \geq v_3^2 + 
    \frac{1 - \norm{\phi}{L^1} + \al^2(1 + \norm{\phi}{L^1})}{2} (v_1^2 + v_2^2)$,
i.e., \eqref{lowerbound} holds, if $2 \be := 1 - \norm{\phi}{L^1} + \al^2(1 + \norm{\phi}{L^1})$ is positive, which in turn is equivalent to requiring 
\begin{equation}\label{alphaphi}
  \al^2 > \frac{\norm{\phi}{L^1} - 1}{\norm{\phi}{L^1} + 1} =: c_\phi.
\end{equation}
Note that $0 \leq c_\phi < 1$, since $\norm{\phi}{L^1} \geq 1$. We have several regularization policies to guarantee \eqref{alphaphi}, and thereby also \eqref{lowerbound}:

\begin{trivlist}
\item[Variant A:] Choosing $\phi$ nonnegative we obtain $c_\phi = 0$, hence \eqref{lowerbound} holds with $\be = \al^2$. But then we cannot require the higher moment conditions $\d_1^l \d_2^k \FT{\phi} (0) = 0$ ($l,k \geq 1$), thus we drop the strict equality of smooth functions with their (classes of) $\phi$-convolved regularizations within the algebra of generalized functions.

\item[Variant B:] We choose $\phi$ in the standard way, i.e., including the higher moment conditions, but obtain the above kind of metric splitting (with $\be = \frac{\al^2 - c_\phi}{2(\norm{\phi}{L^1} + 1)}$) only for restricted values of $\al$ in the interval $]\sqrt{c_\phi},1]$. Or, in partial reverse, if $\al > 0$ is given, then we may choose a mollifier $\phi$ (including the higher moment conditions) with $\norm{\phi}{L^1} - 1$ so small that $c_\phi < \al^2$ holds (the possibility of this follows from more general results concerning constructions of one-dimensional mollifiers with prescribed moments in \cite{Spreitzer:14} and taking tensor products to obtain higher dimensional variants). 

\item[Variant C:] Instead of using a model delta regularization based on scaling a single mollifier $\phi$ we could follow the smoothing method introduced in \cite[Lemma 4.3]{SV:09} by employing a strict delta net $(\psi_\eps)$ with $\norm{\psi_\eps}{L^1} \to 1$ (as $\eps \to 0$) and putting $f_j^\eps = f_j * \psi_\eps$; then the inequality analogous to \eqref{alphaphi} (with $\psi_\eps$ in place of $\phi$) will be valid for any $\al$, if $\eps$ is sufficiently small.

\end{trivlist}

\begin{remark} On first impression one might expect that a discussion of condition \eqref{alphaphi} would be avoidable by employing the full Colombeau algebra instead of the special version used here. Then (at least with the algebra on $\R^4$) essentially the mollifier sets $A_0, A_1, \ldots$, where those in $A_p$ satisfy moment conditions up to order $p$, take over the role of the parameter $\eps$; however, as far as we could see, we would still run into an estimation argument very similar to the above and end up with the inequality in \eqref{alphaphi}, unless only nonnegative mollifiers from $A_p$ will be considered. 
\end{remark}

\section{Cauchy hypersurface in a generalized sense}

The conical spacetime appears in \cite{ES:77} as an example showing a  \emph{quasi-regular singularity} along the two-dimensional submanifold where $(x,y) = (0,0)$, which cannot be detected by unboundedness of curvature components (measured in a parallely propagated frame) near the singularity; the singularity was identified as one of a topological nature, so that the spacetime cannot be extended into that region in a globally consistent way. Thus Ellis and Schmidt concluded that ``there can be no global Cauchy surfaces in these space-times''. James Vickers later discussed the extendability of the spacetime in a distribution theoretic setting and in a wider geometric context  in \cite{Vickers:87,Vickers:90}, prior to moving on to a full analysis in the framework of generalized functions starting with \cite{CVW:96} and with fundamental results on the Cauchy problem for the wave equation in \cite{VW:00}, which implicitly indicate that it should be possible to consider the initial value hypersurface $\{ t = 0\}$ a Cauchy hypersurface in an appropriately generalized sense. In fact, this was one of Clarke's major points of motivation in establishing $\Box$-global hyperbolicity in \cite{Clarke:98}: Several models with non-smooth metrics are known in general relativity that violate the cosmic censorship hypothesis -- hence are not globally hyperbolic -- due to necessary deletion of points of non-smoothness from the universe in order to save the usual geometric methods (based on geodesics etc), despite the fact that one can achieve global well-posedness of the Cauchy problems for the physical fields in a weaker sense; still, the initial value surface then cannot be a Cauchy hypersurface in the classical sense, since Geroch proved in \cite{Geroch:70} that  global hyperbolicity of a spacetime $(M,g)$ is equivalent to the existence of a Cauchy hypersurface $S$.

A standard definition of a \emph{Cauchy hypersurface} $S \subset M$ requires that  every inextendible  timelike curve intersects $S$ exactly once (\cite[Chapter 14, Definition 28]{ONeill:83}).  In the context of differential geometry with generalized objects, such as discontinuous vector and tensor fields, a difficulty arises as how to define the notion of an \emph{inextendible generalized curve}, because therein the required non-existence of limits toward the interval boundaries seems to have no natural generalization. However, there might be an alternative approach by remembering that for a $C^1$ curve defined as map on an open interval $I$ into the manifold $M$ one can resort to an arclength parametrization with respect to a complete Riemannian metric, more precisely:  Choose  a complete Riemannian metric $\eta$ on $M$ and let $\ga \col I \to M$ be $C^1$. Pick $t_0 \in I$, then the arclength parameter $s(t) := \int_{t_0}^t \sqrt{\eta(\ga(\tau))(\dot\ga(\tau),\dot\ga(\tau))} \, d\tau$ (giving negative values, if $t < t_0$) defines a continuous nondecreasing map $s \col I \to \R$ such that 
\beq\label{inextendible}
  \ga \text{ is inextendible if and only if } s(I) = \R.
\eeq
\begin{proof} If $s(I) \not= \R$, then $s$ is bounded above or below near one of the interval boundaries of $I$ and, by monotonicity, $s(t_n)$ has a limit as $t_n \in I$ approaches that boundary; in particular, $(s(t_n))$ is a Cauchy sequence in $\R$. Let $d$ denote the distance function on $M$ induced by the complete Riemann metric $\eta$, then $d(\ga(t_n),\ga(t_m)) \leq |s(t_n) - s(t_m)|$ (length of the corresponding curve segment), showing that $(\ga(t_n))$ is a Cauchy sequence in the metric space $(M,d)$. Hence $\ga$ is extendible. 
On the other hand, if $\ga$ is extendible, say the limit of $\ga(t)$ exists as $t$ approaches the right boundary point $r$, then $s_{[t_0,r[}$ has finite length and $s$ is bounded above, hence $s(I) \not= \R$. \end{proof}

Thus, among all regular $C^1$ curves, i.e., $C^1$ maps $\la\col I \to M$ such that $\dot{\la}(t) \neq 0$ for all $t \in I$, precisely the inextendible curves possess a re\-para\-metrization in the form $\ga \col \R \to M$ with $\eta(\ga(s))(\dot\ga(s),\dot\ga(s)) = 1$ for all $s \in \R$. We will pick this as a starting point for our first approach to the notion of inextendibility in the generalized sense, but we refrain from turning this into an official definition, because its  dependence on the allowed class of generalized reparametrizations is troublesome. 

As for the notion of a \emph{timelike generalized curve} with respect to a generalized Lorentz metric, the natural definition clearly would be be to require strict negativity of $g(\dot{\ga},\dot{\ga})$ as generalized function on an interval. However, for technical reasons in the proof following below we strengthen this by a uniformity requirement.  As above, we do see this only as a preliminary concept (because it is parametrization dependent) not justifying a definition of such notion yet. 

More systematic and deeper attempts to investigate appropriate notions of a generalized Cauchy hypersurface for generalized spacetimes will remain a future task. As a preliminary observation we may present the following statement, in which $e$ could be replaced by any  complete Riemannian metric, the condition of normalized tangent vectors together with requiring that the generalized map defining the curve has domain $\R$ mimics the inextendibility condition, and a uniform strict negativity property implements a strong variant of the property ``timelike''.

\begin{proposition} Let $(\R^4,g_\al)$ be the generalized conical spacetime, $S = \{0\} \times \R^3$, and $e$ denote the Euclidean metric on $\R^4$. Suppose the generalized curve $\ga \in \G[\R,\R^4]$ is parametrized such that $e(\dot\ga,\dot\ga) = 1$ holds in $\G(\R)$  and $g_\al(\dot\ga,\dot\ga)$ is \emph{uniformly strictly negative}, i.e., for any representatives $(\ga^\eps)_{\eps \in\,]0,1]} \in \ga$ and $(g_\al^\eps)_{\eps \in\,]0,1]} \in g_\al$ there exist $q \in [0,\infty[$ and $\eps_0 \in \,]0,1]$ such that
$$
   g_\al^\eps(\ga_\eps(s))(\dot\ga_\eps(s),\dot\ga_\eps(s)) \leq -\eps^q 
      \qquad  (s \in \R, 0 < \eps < \eps_0).
$$
Then $\ga$ meets $S$ in a unique (compactly supported) generalized point value, i.e., there is a unique $\tilde{s} \in (\tilde{\R})_\text{c}$ such that $\ga(\tilde{s}) \in \tilde{S}_\text{c}$. 
\end{proposition}
\begin{proof} Let $\ga$ be represented by $\ga_\eps = (\ga_0^\eps, \ga_1^\eps) \col \R \to \R\times \R^3$ and $g_\al$ be represented by $g_\al^\eps$. We will also employ the metric splitting established in Section 3, in particular the estimate \eqref{lowerbound}. The property of $\ga$ being uniformly timelike then implies that we have with some $q \in \R$ and $\eps_0 \in\,]0,1]$:
$$
   - \eps^q \geq g_\al^\eps(\dot\ga_\eps,\dot\ga_\eps) = 
     - (\dot\ga_0^\eps)^2 + \rho^\eps(\dot\ga_1^\eps,\dot\ga_1^\eps) 
     \geq - (\dot\ga_0^\eps)^2 + \be \Norm{\dot\ga_1^\eps}^2
     \quad (0 < \eps < \eps_0), 
$$
where $\be$ here denotes the minimum of $1$ and the constant occurring in \eqref{lowerbound}; in particular, 
\beq\tag{$\star$}
    \Norm{\dot\ga_1^\eps(s)}^2 \leq 
      \frac{1}{\be} (\dot\ga_0^\eps(s)^2 -  \eps^q)  \leq \frac{1}{\be} \dot\ga_0^\eps(s)^2
    \qquad \forall s \in \R, 0 < \eps < \eps_0.
\eeq
The parametrization (``inextendibility'') condition for $\ga$ gives $(\dot\ga_0^\eps)^2 + \Norm{\dot\ga_1^\eps}^2 = 1 + n_\eps$ with negligible $(n_\eps)_{\eps \in \,]0,1]}$, and  combination with ($\star$) implies
$$
   1 + n_\eps(s) \leq \dot\ga_0^\eps(s)^2 + \frac{1}{\be} \dot\ga_0^\eps(s)^2  
   = \frac{\be + 1}{\be}\, \dot\ga_0^\eps(s)^2,
$$
hence $\dot\ga_0^\eps(s)^2 \geq \frac{\be}{\be + 1}(1 + n_\eps(s))  \geq \frac{\be}{\be + 1} \cdot \frac{1}{2}$, if $\eps > 0$ is sufficiently small; in summary, there is some $\eps_1 \in \,]0,1]$ such that
\beq\tag{$\star\star$}
   \dot\ga_0^\eps(s)^2 \geq \frac{\be}{2(\be + 1)}
    \qquad \forall s \in \R, 0 < \eps < \eps_1.
\eeq
This implies that $\ga_0^\eps \col \R \to \R$ is bijective, if $0 < \eps < \eps_1$. Let $s_\eps$ denote the unique point in $\R$ such that $\ga_0^\eps(s_\eps) = 0$ and put $p_\eps := \ga_\eps(s_\eps) = (0,\ga_1^\eps(s_\eps)) \in S$.

To see that $(s_\eps)$ is c-bounded we write 
$$
  s_\eps = (\ga_0^\eps)^{-1}(0) = \int_{\ga_0^\eps(0)}^0  
    \diff{\tau} (\ga_0^\eps)^{-1}(\tau) \, d\tau =
    \int_{\ga_0^\eps(0)}^0   \frac{d\tau}{\dot\ga_0^\eps((\ga_0^\eps)^{-1}(\tau))} 
$$
and note that ($\star\star$) implies  $1/|\dot\ga_0^\eps| \leq \sqrt{2(\be + 1)/\be}$, thus   $|s_\eps| \leq |\ga_0^\eps(0)| \sqrt{2(\be + 1)/\be}$ and we may call on the c-boundedness of $\ga_0$.  Thus $(s_\eps)$ defines a compactly supported generalized point $\tilde{s} \in (\widetilde{\R})_{\text{c}}$ and  $\tilde{p} = \ga(\tilde{s})$ in the sense of generalized point values with $\tilde{p}$ represented by $(p_\eps)$. Since $\ga$ is c-bounded the generalized point value  $\tilde{p} = \ga(\tilde{s})$ belongs to $\widetilde{S}_{\text{c}}$.

To prove uniqueness let us assume that $\tilde{q} = \ga(\tilde{t}) \in \widetilde{S}$ with $\tilde{t} \in (\widetilde{\R})_{\text{c}}$. We have $\ga_0(\tilde{t}) = 0 = \ga_0(\tilde{s})$, hence $0 = \ga_0(\tilde{t}) - \ga_0(\tilde{s}) = \int_{\tilde{s}}^{\tilde{t}} \dot\ga_0(\tau)\, d\tau$ (the integral understood $\eps$-wise on representatives, including the upper and lower limit). By ($\star\star$) we have that $\dot\ga_0$ is strictly nonzero (even bounded away from zero by a classical constant), therefore it is impossible to find any interval of length with non-negligible asymptotics between $s_\eps$ and $t_\eps$, hence  $\tilde{t} = \tilde{s}$ (and $\tilde{q} = \tilde{p}$).
\end{proof}

\begin{example} (i) The  generalized curve $\ga \in \G[\R,\R^4]$ with representative given by $\ga_\eps(s) = (s, \sin(\frac{1}{\eps}),0,0)$ satisfies the hypothesis of the proposition. It meets the ``Cauchy hypersurface'' $S$ in a generalized point with accumulation points forming the subset $\{0\} \times [-1,1]\times \{(0,0)\}$. 

\noindent (ii) The Colombeau-generalized map $\la \in \G(\R,\R^4)$ represented by $\lambda_\eps(s) = (s,\frac{1}{\eps},0,0)$ is not c-bounded, hence not considered a \emph{generalized curve}, though it satisfies the other conditions in the above proposition.  The ``image'' of $\la$ meets $\widetilde{S}$ in a generalized point which is not compactly supported and has no accumulation points of corresponding representative nets.
\end{example}

\section{Causal curves and approximations}

\subsubsection*{What is a generalized causal curve?} A notion of  causal curve in the generalized sense seems delicate for two reasons: First, a condition like $g(\dot\ga,\dot\ga) \leq 0$ in the generalized function sense would not guarantee to have $g_\eps$-causal representatives $\ga_\eps$ of $\ga$; second, the existence of zero divisors in the set of generalized numbers raises the question whether a nonzero, but non free, tangent vector $v$ with $g(v,v) = 0$ should be considered lightlike, spacelike, or neither. 
Even requiring regularizations to consist of causal curves brings us to the practical question of approximation. It could have a variety of meanings, at least  two of which one would hope to be essentially equivalent in fairly general classical situations: Causal curves having image sets within a given open set, or causal curves being uniformly close as mappings from intervals into the manifold. Well-known results on these questions address the sequential convergence in strongly causal Lorentz manifolds or abstract order theoretic constructions for the globally hyperbolic case. We want to use the opportunity here to a provide (or give a detailed review of) a precise statement, in what sense the versions using classes of curves and specific parametrizations are topologically comparable. 

\subsubsection*{Two topologies for continuous causal curves:} A  generalized Lorentz metric satisfying the globally hyperbolic metric splitting according to Definition \ref{G_metric_splitting} possesses representatives consisting of smooth globally hyperbolic Lorentz metrics. Thus, in these regularizations we may take advantage of very good causality properties, in particular, including strong causality etc. As discussed, e.g., in the review article \cite{Sanchez:11} by S\'{a}nchez global hyperbolicity of a smooth spacetime $(M,g)$ has been described in \cite{C-B:68} in terms of compactness with respect to the compact-open topology of the set of causal curves connecting two points and parametrized proportionally to arclength on the interval $[0,1]$, whereas more often the so-called $C^0$-topology on the set of equivalence classes or images of causal curves  is being employed instead (e.g.\ in \cite{Penrose:72, Wald:84, Kriele:99}). We investigate the relation between the corresponding topological spaces of causal curves or classes of causal curves. We may refer to \cite{Saemann:14} (and the references therein) for more details on the analysis of comparability of several related topologies for Lorentzian metrics that are merely \emph{continuous}.

Let $h$ be a complete Riemannian metric with induced metric $d_h$ on $M$ and denote the compact-open topology on $C([0,1],M)$ (and subsets thereof)  by $\tauco$. It is metrizable by the metric  $\rho$ of uniform convergence $\rho(\la,\ga) = \sup \{ d_h(\la(s),\ga(s)) \mid s \in [0,1]\}$. Let $p, q \in M$ satisfy $p < q$ in the sense of the causal relation (as in \cite[Chapter 14]{ONeill:83}); $p \leq q$ means $p < q$ or $p = q$.

Recall that a \emph{continuous} curve $\ga \col [0,1] \to M$ is said to be \emph{causal}, 
if for any  convex open subset $U$ of $M$ and $s_1, s_2 \in [0,1]$, $s_1 \leq s_2$ with $\ga([s_1,s_2]) \subset U$, we have  $\ga(s_1) \leq \ga(s_2)$ (relative $U$). Note that limits of continuous causal curves in the sense of the compact-open topology are continuous causal curves. In fact, let $\ga \in C([0,1],M)$ be the uniform limit of the sequence  of continuous causal curves $\ga_n \col [0,1] \to M$. For large $n$ we have $\ga_n([s_1,s_2]) \subset U$ and clearly $\ga_n(s_j) \to \ga(s_j)$ ($n \to \infty$; $j = 1,2$). Since $\ga_n$ is causal we have $\ga_n(s_1) \leq \ga_n(s_2)$ and the property that the relation $\leq$ is closed on $U$ (cf.\ \cite[Chapter 14, Lemma 2]{ONeill:83}) implies that $\ga(s_1) \leq \ga(s_2)$ relative $U$. Hence $\ga$ is causal.

Continuous causal curves are \emph{Lipschitz continuous} (see, e.g., pages 75-76 in \cite{BEE:96}, or \cite{Kriele:99}, pages 365-366) and therefore rectifiable and  differentiable almost everywhere, in particular, they may be parametrized by arclength. Moreover, if the image sets of all curves in a subset $A$ of $C([0,1],M)$ consisting of continuous causal curves are contained in a fixed compact set in $M$, then we even have a uniform Lipschitz constant $L$ for every curve in $A$ as well as a uniform upper bound for the curve lengths. The uniform Lipschitz constant implies \emph{equicontinuity} of $A$, since for every $\ga \in A$ and $s_1, s_2 \in [0,1]$ we have  $d_h(\ga(s_1),\ga(s_2)) \leq L |s_2 - s_1|$, which is arbitrarily small independently of $\ga$, if the difference $|s_2 - s_1|$ is sufficiently small.

Let $C(p,q)$ be the set of all \emph{classes of continuous causal curves} from $p$ to $q$ modulo continuously differentiable parameter transforms with positive derivative. The following property is  mentioned in \cite[page 206]{Wald:84} (and discussed also in \cite{Saemann:14}):
\beq \tag{$\star$}
  (M,g) \text{ causal},\; \widetilde{\la}, \widetilde{\ga} \in C(p,q):
  \quad \text{(image of) } \widetilde{\la} \subseteq 
       \text{(image of) } \widetilde{\ga} 
       \quad \Longrightarrow \quad 
       \widetilde{\la} = \widetilde{\ga}.
\eeq
Thus, we may identify classes of curves in $C(p,q)$ with their images as subsets of $M$ and it is easy to show that we obtain a basis for a Hausdorff topology $\tau$ on $C(p,q)$ by specifying $O(U) := \{ \widetilde{\la} \in C(p,q) \mid \text{(image of) } \widetilde{\la} \subset U \}$ for any open $U \subseteq M$.

Let  $C_\co(p,q)$  be the $\tauco$-closure of the set of all future directed continuous causal curves $\la \in C([0,1],M)$ such that $\la(0) = p$, $\la(1) = q$, and  $h(\dot\la,\dot\la)$ is constant almost everywhere.

Note that every class of a continuous causal curve $\sig_0$ from $p$ to $q$ possesses a unique parametrization $\sig$ defined on $[0,1]$ and proportional to arclength, i.e., with a constant $l > 0$ such that $h_{\sig(s)}(\dot\sig(s),\dot\sig(s)) = l^2$ for almost every $s \in [0,1]$:  As noted above $\sig_0$ is rectifiable and differentiable a.e., hence possesses a parametrization by arclength; upon rescaling by the length $l$ we obtain the desired parametrization on the interval $[0,1]$. As for uniqueness, suppose $\ga \col [0,1] \to M$ is another such parametrization in the same class. If $\vphi \col [0,1] \to [0,1]$ is $C^1$ with $\vphi' > 0$ and such that $\sig = \ga \circ \vphi$, then $l = h(\dot{\sig},\dot{\sig}) = (\vphi')^2 \cdot h(\dot{\ga}, \dot{\ga}) =  (\vphi')^2 \cdot l$ holds a.e.; hence $\vphi' = 1$ and therefore $\sig = \ga$.

\begin{proposition}  Then the map $\Ga \col (C_\co(p,q),\tauco) \to (C(p,q),\tau)$, $\Gamma(\gamma) := \widetilde{\ga}$, the class of $\ga$ in $C(p,q)$, is surjective, continuous, and proper. In particular, compactness of $(C_\co(p,q),\tauco)$ is implied by that of $(C(p,q),\tau)$ and vice versa.
\end{proposition}

\begin{proof} The map $\Ga$ is surjective by the observation on existence of parametrizations on $[0,1]$ proportional to arclength made above. (Remark: Dropping the closure in the definition of $C_\co(p,q)$ would yield injectivity of $\Ga$ as well.)

The map $\Ga$ is $\tauco$-$\tau$-continuous:  Let $\ga \in C_\co(p,q)$ and $O(U)$ be an open base neighborhood of $\Ga(\ga)$. Choose $\eps > 0$ sufficiently small such that $(\ga)_\eps := \{ x \in M \mid d_h(x, \text{image of } \ga) < \eps\}$ is contained in $U$, then $\Ga(B_\eps^\rho(\ga)) \subseteq O(U)$.

It remains to prove that preimages of compact subsets in $(C(p,q),\tau)$ are compact in $(C_\co(p,q),\tauco)$. A compact subset $K$ of $C(p,q)$ can be covered by finitely many open basis sets of the form $O((\widetilde{\ga})_1)$ with $\widetilde{\ga}$ belonging to $K$. The $d_h$-bounded subset $(\widetilde{\ga})_1$ is open and relatively compact in $M$, since $h$ is a complete Riemannian metric. 
Hence it suffices to prove the the following

\noindent \emph{Claim:} If $U \subseteq M$ is open and relatively compact, then $\Ga^{-1}(O(U))$ is relatively compact in $C_\co(p,q)$.

\noindent Proof of the Claim: All continuous causal curves $[0,1] \to M$ with images contained in $U$ are Lipschitz continuous with a uniformly bounded Lipschitz constant $L$ and hence  $\Ga^{-1}(O(U))$ is equicontinuous by an observation made earlier. By $d_h$-boundedness of $U$, we also get that $\Ga^{-1}(O(U))$ is pointwise bounded. Thus, the Arzela-Ascoli theorem implies that $\Ga^{-1}(O(U))$ is relatively compact with respect to $\tauco$. 
\end{proof}

 \paragraph{\textbf{Acknowledgement:}} The author thanks an anonymous referee for suggestions to improve a few formulations and aspects in Sections 1-4, and thanks Clemens S\"amann for corrections and helpful remarks on Section 5. Work on this paper has been supported by the Austrian Science Fund project P25326.

\bibliography{gh}

\end{document}